%% file: main.tex
\documentclass{article}
\usepackage{arxiv}
\usepackage{amssymb}
\usepackage{amsmath}
\usepackage{amsthm}
\usepackage{hyperref}
\usepackage{bbm}
\usepackage{comment}

\usepackage{graphicx}
\usepackage[compact]{titlesec}
\usepackage[ruled,vlined]{algorithm2e}
\usepackage{booktabs}
\usepackage{enumitem}

\usepackage{xfrac}
\usepackage{xcolor}

\usepackage{natbib}
\AtBeginDocument{%
  \providecommand\BibTeX{{%
    \normalfont B\kern-0.5em{\scshape i\kern-0.25em b}\kern-0.8em\TeX}}}

\newtheorem{assumption}{Assumption}

\newtheorem{proposition}{Proposition}

\title{Efficient Balanced Treatment Assignments for Experimentation}

\author{
David Arbour\\
Adobe Research\\
\texttt{darbour26@gmail.com}
\And
Drew Dimmery\\
Facebook Core Data Science\\
\texttt{drewd@fb.com}
\And 
Anup Rao\\
Adobe Research\\
\texttt{anuprao@adobe.com}
}

\begin{document}

\maketitle

\begin{abstract}
In this work, we reframe the problem of balanced treatment assignment as optimization of a two-sample test between test and control units. 
Using this lens we provide an assignment algorithm that is optimal with respect to the minimum spanning tree test of~\citet{friedman1979multivariate}.
This assignment to treatment groups may be performed \emph{exactly} in polynomial time.
We provide a probabilistic interpretation of this process in terms of the most probable element of designs drawn from a determinantal point process which admits a probabilistic interpretation of the design. 
We provide a novel formulation of estimation as transductive inference and show how the tree structures used in design can also be used in an adjustment estimator.
We conclude with a simulation study demonstrating the improved efficacy of our method.
\end{abstract}

\input{maincontent.tex}
\bibliographystyle{unsrtnat}
\bibliography{references}
\appendix
\onecolumn

\section{Proofs}
\begin{proof}[Proof of Proposition~\ref{prop:knnbound}\citep{anava2016k}]
\begin{align}
\nonumber
&\left|\sum_{i=1}^{n} w_{\star, i} y_{i}-f\left(x_\star\right)\right|\\ 
\nonumber
&=\left|\sum_{i=1}^{n} \alpha_{i}\left(y_{i}-f\left(x_{i}\right)+f\left(x_{i}\right)\right)-f\left(x_{0}\right)\right| \\
\nonumber
&=\left|\sum_{i=1}^{n} w_{\star, i} \epsilon_{i}+\sum_{i=1}^{n} w_{\star, i}\left(f\left(x_{i}\right)-f\left(x_{0}\right)\right)\right| \\
\nonumber
& \leq\left|\sum_{i=1}^{n} w_{\star, i} \epsilon_{i}\right|+\left|\sum_{i=1}^{n} w_{\star, i}\left(f\left(x_{i}\right)-f\left(x_{\star}\right)\right)\right| \\
\label{app:eq:knnbound}
& \leq \underbrace{\left|\sum_{i=1}^{n} w_{\star, i} \epsilon_{i}\right|}_{\text{variance}} + \underbrace{L\sum_{i=1}^{n} w_{\star, i} d\left(x_{i}, x_{\star}\right)}_{\text{bias}}
\end{align}
Where $L$ is the Lipschitz constant as in assumption \ref{ass:SMOOTH}, and $d$ is a distance measure.
By placing an additional assumption that the noise term can be bounded by a constant, $|\epsilon| \leq b$, the variance term can be further bounded with probability $1 - \delta$ using an application of Hoeffding's inequality~\citep{anava2016k},
\begin{align*}
    \left|\sum_{i=1}^{n} w_{\star, i} \epsilon_{i}\right| \leq C\|\mathbf{w}_\star\|_2,
    C = b\sqrt{2\log(\frac{2}{\delta})}
\end{align*}
\end{proof}
\begin{proof}[Proof of Proposition~\ref{prop:main_bound}]
By triangle inequality, we have for every $i,$
$$
 \left| \frac{\sum_j \mathbbm{1}(a_i\neq a_j) e_{ij}y_j}{\sum_k \mathbbm{1}(a_i\neq a_k) e_{ik}} - y_i \right| \leq 
  \left| \frac{\sum_j \mathbbm{1}(a_i\neq a_j) e_{ij}y_j}{\sum_k \mathbbm{1}(a_i\neq a_k) e_{ik}} -  \frac{\sum_j  e_{ij}y_j}{\sum_k e_{ik}} \right| + \left| y_i -  \frac{\sum_j  e_{ij}y_j}{\sum_k e_{ik}} \right|.
$$
 Let us also denote $d_i = \sum_j e_{ij}$ and $d_i(\text{cut}) = \sum_{j} \mathbbm{1}(a_i\neq a_j) e_{ij}.$
\begin{align*}
\left|\sum_j \frac{e_{ij} y_j}{\sum_k e_{ik}} - \sum_j \frac{\mathbbm{1}(a_i\neq a_j) e_{ij} y_j}{\sum_k \mathbbm{1}(a_i\neq a_j) e_{ik}}\right| & = 
\left|\sum_j \frac{e_{ij} y_j}{d_i} - \sum_{j} \mathbbm{1}(a_i\neq a_j)\frac{ e_{ij} y_j}{d_i(\text{cut})} \right| \\ & = 
\left| 
\sum_{j } \mathbbm{1}(a_i\neq a_j) e_{ij} y_j \left( \frac{1}{d_i} - \frac{1}{d_i(\text{cut})} \right) + \sum_{j} \frac{(1-\mathbbm{1}(a_i\neq a_j)) e_{ij} y_j}{d_i}
\right| \\
& \leq \left | 0 + \sum_{j} \frac{1-\mathbbm{1}(a_i\neq a_j)) e_{ij}}{d_i} \right| \\
& = \sum_{j} \frac{1-\mathbbm{1}(a_i\neq a_j)) e_{ij}}{d_i} \\
&\leq   \frac{d(i) - \sum_{j} \mathbbm{1}(a_i\neq a_j) e_{ij}}{d_{\min}}.
\end{align*}
The first inequality follows by assuming that $0 \leq y_i \leq 1$ and noting that $1/d - 1/d(C) \leq 0 .$ So the first summand is maximized by setting $y_j = 0$ for $j \in C$ and second one is maximized by setting $y_j = 1$ for $j \in D.$ The second inequality follows by $d \geq d_{min}$. Finally, by summing over all $i$, we have 
$$ \sum_i \left| \frac{\sum_j \mathbbm{1}(a_i\neq a_j) e_{ij}y_j}{\sum_k \mathbbm{1}(a_i\neq a_k) e_{ik}} - y_i \right| \leq  \sum_i \left| y_i -  \frac{\sum_j  e_{ij}y_j}{\sum_k e_{ik}} \right| +  \frac{e_{\text{sum}}-\sum_{i,j} \mathbbm{1}(a_i\neq a_j) e_{ij}}{d_{\min}} .
$$
This is minimized by the max-cut. 
\end{proof}

\subsection{\citet{kallus2018optimal} and Maxcut}

Given the graph, an equivalent formulation of Equation~\ref{prob:PSOD} is finding the weighted maximum cut on the graph~(find a subset of the vertices such that the total weight of edges connecting nodes of the two different subsets is maximized). 
The equivalence is made plain by considering the binary quadratic program formulation of Maxcut given in Equation~\ref{eq:qp_maxcut},
\begin{align}
    \nonumber
    &\arg\max_{\mathbf{u} \in \{-1, 1\}} \mathbf{u}^T L \mathbf{u} = \arg\max_{\mathbf{u} \in \{-1, 1\}} \mathbf{u}^TD\mathbf{u} - \mathbf{u}^T G\mathbf{u}  \\
    \label{prob:maxcut}
    &= \arg\max_{\mathbf{u} \in \{-1, 1\}} - \mathbf{u}^T G\mathbf{u}
\end{align}
where we have defined $L$ to be the combinatorial graph Laplacian, $D - G$ where $D$ is a diagonal matrix where $D_i$ is the degree of vertex $i$ and $G$ is the weighted adjacency matrix.
Comparing this to problem \ref{prob:PSOD}, 
\begin{align*}
    &\arg\min_{\mathbf{u} \in \{-1, 1\}} \mathbf{u}^T K \mathbf{u} = \arg\min_{\mathbf{u} \in \{-1, 1\}} \mathbf{u}^T \text{diag}(K)\mathbf{u} + \mathbf{u}^T G \mathbf{u} \\
    &= \arg\min_{\mathbf{u} \in \{-1, 1\}} \mathbf{u}^T G\mathbf{u}
\end{align*}
we can see that problem \ref{prob:maxcut} given by maxcut is isomorphic to problem \ref{prob:PSOD}, given by \citet{kallus2017balanced}'s PSOD strategy. 
Therefore, improved approximations to Maxcut will additionally be improved approximations to \citet{kallus2018optimal}.

\subsection{The kernel objective of \citet{kallus2018optimal} as uncentered Maximum Mean Discrepancy}
The kernel objective of \citet{kallus2018optimal} is defined as
\begin{align}
    \label{eq:kallusQP}
    \min_{\mathbf{u} \in \{-1, 1\}} \mathbf{u}^t K \mathbf{u}
\end{align}
where $K$ is the Gram matrix for some reproducing kernel. 
By using the cyclic properties of the trace, we can rewrite the objective in equation \ref{eq:kallusQP} is equivalent to 
\begin{align}
    \label{eq:kallusQPTrace}
    \min_{\mathbf{u} \in \{-1, 1\}} \text{trace}\left(K \mathbf{u}\mathbf{u}^t \right)
\end{align}
a biased estimator of the Hilbert-Schmidt independence criterion with respect to $K$ and the kernel given by  $\mathbf{u}\mathbf{u}^T$ can be written as~\citep{gretton2008kernel}
\begin{align}
    \label{eq:binHSIC}
    \text{trace}(KH\mathbf{u}\mathbf{u}^TH)\\
    \nonumber
    H = I - \frac{1}{n}\mathbf{1}\mathbf{1}^T
\end{align}
Equation \ref{eq:binHSIC}, in turn, was shown to be equivalent to the biased estimator of the Maximum Mean Discrepancy by \citet{song2008learning}.

\section{Simulated data generating processes}

\begin{table*}[h]
\centering\footnotesize
\begin{tabular}{rcccc}
    \toprule
    \textbf{DGP} & $\mathbf{X}$ & $y(0)$ & $y(1)$\\
    \midrule
    \textbf{LinearDGP} & $X_k = \epsilon_{k}, k \in \{1,\dots,4\}$& $\mathbf{X} \beta + \frac{1}{10}\epsilon_{y(0)}$ & 1 + $\mathbf{X} \beta + \frac{1}{10}\epsilon_{y(1)}$\\
    \textbf{QuickBlockDGP} & $X_k \sim \mathcal{U}(0,10), \forall k \in \{1,2\}$ &  $\prod_{k=1}^2 X_{k} + \epsilon$ & 1 + y(0) \\
    \textbf{SinusoidalDGP} & $X_k = \epsilon_{k}, k \in \{1,\dots,4\}$ & $\sin(\mathbf{X} \beta) + \frac{\epsilon_{y(0)}}{10}$ & 1 + $\sin(\mathbf{X} \beta) + \frac{\epsilon_{y(1)}}{10}$\\
    \textbf{TwoCircles} & latent: $r \sim \mathcal{N}(1 + i \% 2,\frac{1}{10})$ & $\beta_1 s + \beta_2 r + \epsilon_{y(0)}$ & $\beta_1 s + \beta_2 r + \epsilon_{y(1)}$\\
    &$s \sim \mathcal{U}(0,2\pi)$&&&\\
    & observed: $x_1 = r \cos(s)$, & &\\
    &$x_2 = r \sin(s)$&&&\\
    \bottomrule
\end{tabular}
\caption{Data generating processes used in simulations. All $\epsilon$s indicate a standard normal variate and all $\beta$s indicate a standard uniform variate. $\%$ indicates the modulo function, and $i$ indicates a unit's index.}
\label{tab:dgps}
\end{table*}

\section{Estimators for Experiments}
For the estimation of ATEs for Bernoulli and rerandomization, we use regression-adjusted estimators as used in \citet{lin2013agnostic}: a linear regression with covariates mean-centered and interacted with treatment.
QuickBlock uses a blocking estimator as the authors propose, and the \citet{kallus2018optimal} designs use a difference-in-means estimator as proposed.
The matched-pairs design takes the average in within-pair outcomes, which leads to a more efficient estimator than difference-in-means~\citep{imai2008variance}.
When examining ITE estimators, we use random forest based T-learners unless otherwise noted~\citep{athey2016recursive, kunzel2019metalearners}.
All methods use the same hyperparameters, with the number of trees set at $20 \times n^{\frac{1}{4}}$ to ensure model complexity grows with sample size and with maximum tree depth set at 8.

\section{Sensitivity to Hyperparameters}

\begin{figure}
    \centering
    \includegraphics[width=1\textwidth]{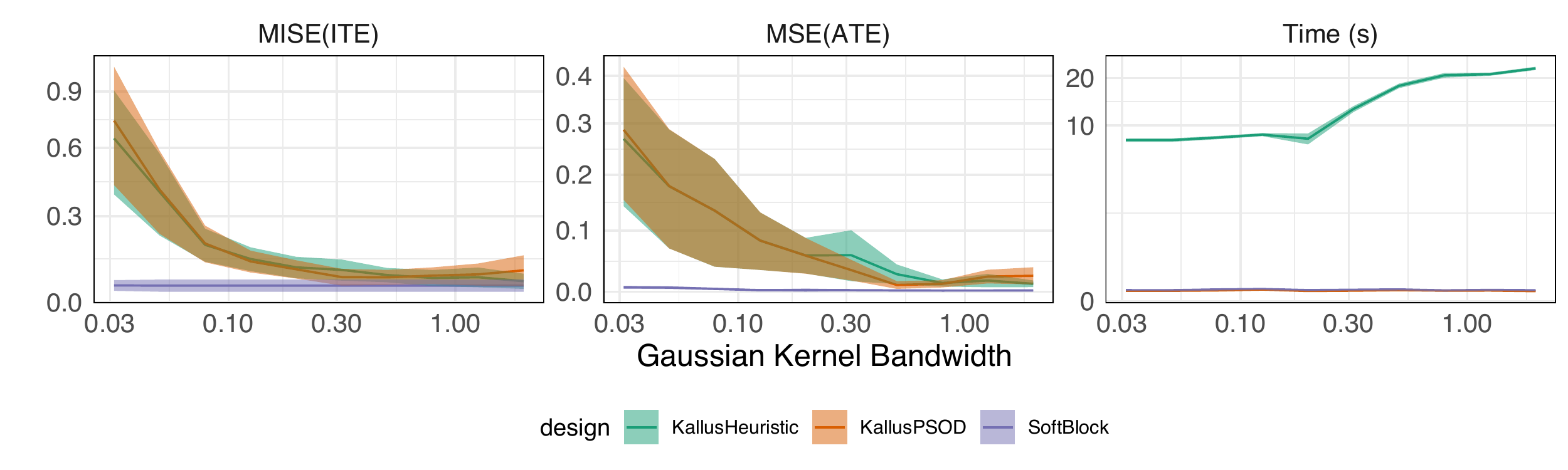}
    \caption{
    SoftBlock is robust to hyperparameter selection.
    This figure shows mean squared-error of the ATE, mean integrated-squared-error of the ITE and time to calculate a design and perform estimation.
    }
    \label{fig:hyperparams}
\end{figure}

Figure~\ref{fig:hyperparams} shows the sensitivity of the \citet{kallus2018optimal} methods and SoftBlock to hyperparameters.
Since both methods are based on similarities defined by a kernel matrix, we plot the performance of these methods on the TwoCircles problem as the bandwidth of the Gaussian kernel changes.
Softblock is not at all sensitive to hyperparameters, performing well at all values, while the \citet{kallus2018optimal} methods perform well only when the hyper-parameters are set well.
In essence, these methods perform covariate adjustment a-priori, but this means that they implicitly specify an outcome model \emph{before data is observed}.
As such, it is very difficult to set these values effectively in practice, as it amounts to tuning a non-parametric model without data for cross-validation or other model selection techniques.

\end{document}

%% file: maincontent.tex
\section{Introduction}
\label{sec:intro}

Decision-making often requires engaging with counterfactual questions.
For instance, determining whether to give a patient a medication depends on what their health outcomes \emph{would have been} absent the medication.
One of the most successful tools for answering these types of counterfactual questions has been experimentation.
For a sample of patients, randomly give half of them the medication and half of them a placebo, and measure the \emph{average} health outcomes for each of the two groups.
This provides unbiased estimates of the typical response in the sample: the average treatment effect (ATE)~\citep{imbens2015causal}.
This does not address a doctor's most fundamental concern, however: how would \emph{this} patient respond to treatment, relative to their counterfactual health outcomes under placebo?
To answer this question, it is necessary to consider \emph{individual treatment effects} (ITE)~\citep{shalit2017estimating}.
While the literature has provided many improvements around the design of experiments to measure the former quantity (the ATE), in this paper, we analyze the problem of experimental design for estimation of the ITE.

This work is concerned with extending the capabilities of experimental design along two axes:

\textit{Experimental Design for Individual Treatment Effects}. To our knowledge, this is the first work focused on design-based solutions to the estimation of ITEs.

\textit{Computationally Efficient Exact Solutions}. Both mean (and kernel mean) based measures of imbalance and blocking are NP-hard to optimize~\citep{kallus2018optimal,higgins2016improving}.

Our primary contributions are:
\begin{itemize}[nosep,leftmargin=1em,labelwidth=*,align=left] 
    \item Motivate the estimation of ITEs around transductive learning.
    \item Show that the problem of good experimental design is closely related to a ubiquitous graph-cutting problem through a bias-variance decomposition of the design problem.
    \item Reorient the problem of balance around a two-sample test between treatment and control covariate profiles.
    \item Provide an efficient approximation to this problem based on maximum spanning trees, which optimizes a ubiquitous graph-based two-sample test and provides highly accurate estimates of ITEs.
\end{itemize}

The structure of this paper is as follows.
Section~\ref{sec:background} describes the problem of experimental design, and estimation of ITEs given a design.
Section~\ref{sec:related_work} provides an overview of pre-existing work on experimental design.
Section~\ref{sec:ite_design} presents the problem of ITE-optimizing experimental design, connects it to graph cutting and discusses existing approaches through this lens.
Section~\ref{sec:novel_design} presents our proposed design which optimizes test of balance based on the minimum spanning tree.
Section~\ref{sec:experiments} shows a bevy of simulation evidence demonstrating the strength of our proposed design.

\section{Background and Problem Description}
\label{sec:background}
We first give some background and notation before introducing the task of this work. 
Throughout we will consider three sets of variables, $\mathbf{X} \in \mathbb{R}^D$, $A \in \{0, 1\}$ and $Y \in \mathbb{R}$. 
We assume that $\mathbf{X}$ is pretreatment, i.e. the values are not caused by $A$ or $Y$. 
We will also assume that $Y$ is given as some function of $X$, $A$, and mean-zero noise.
Given a set $i =1, \dots, N$ of realizations of $Y$ and $A$ the potential outcomes~\citep{rubin2011causal}, $Y^{A=0}$, $Y^{A=1}$ are the values of $Y$ that would have been observed had treatment been observed at $A=1$ or $A=0$, respectively.
For some mathematical statements it is more convenient to annotate treatment as being in $\{-1, 1\}$, and we will indicate this by the use of a vector $\mathbf{u}$, wherein $-1$ notates control and $1$ indicates treatment.
Causal effects are then, in turn, derived as contrasts between potential outcomes. 
In this paper, we will consider two causal estimands:
\begin{itemize}[nosep,leftmargin=1em,labelwidth=*,align=left] 
\item \text The \emph{Individual Treatment Effect~(ITE)} is the conditional effect of treatment, $\textrm{ITE} = \mathbb{E}\left[Y^{A=1} - Y^{A=0} | \mathbf{X} = \mathbf{x} \right]$.
\item The \emph{Average Treatment Effect~(ATE)} is an estimate of the marginal effect of treatment from a finite sample. The ATE is easily expressed as an expectation of the ITE,  $\int_{\mathbf{x}\in\mathbf{X}} \mathbb{E}\left[Y^{A=1} - Y^{A=0} | \mathbf{X} = \mathbf{x} \right] p(\mathbf{x}) d\mathbf{x}$.
\end{itemize}

Optimal experimental design---the central task of this paper---considers the following problem. 
Given the set of pre-treatment covariates $\mathbf{X}$, how should treatment be assigned to each individual in order to to obtain an unbiased estimate of a causal estimand with minimal variance?
Optimal experimental design for estimating the ATE has been studied for decades~(c.f., \citet{fisher1935design, morgan2012rerandomization, hall1995blocking, kallus2018optimal, higgins2016improving}).
In the general setting, \citet{kallus2018optimal} showed that complete randomization is minimax optimal.
However, with additional assumptions placed over the potential outcomes, improvement can be made through careful allocation. 
One such assumption, which we will employ throughout the remainder of the paper, is that the potential outcomes are smooth functions with additive noise.
More precisely, we introduce the following assumptions
\begin{assumption}
\label{ass:covariates}
The pre-treatment covariates, $\mathbf{x}$, belong to a metric space, with the corresponding metric denoted $d(x, x')$, and are drawn from some distribution $p(x)$ with finite variance. 
In this paper, we assume that $x_i$ is drawn from some (possibly unknown) distribution $p(x)$ and that the domain is a metric space, $\mathcal{X}$ with the metric $d_X$.
\end{assumption}
\begin{assumption}
\label{ass:additive}
Each of the potential outcomes, are drawn from the following generative process
\begin{align*}
    &f^{A=a}(\mathbf{X}_i) = \mathbb{E}\left[Y^{A=a}_i | \mathbf{X}_i = x \right]\\
    &\epsilon^{A=a}_i = Y^{A=a}_i - f^{A=a}(\mathbf{X}_i)
\end{align*}
Where $\epsilon^{A=a}$ is mean zero.%
\end{assumption}
\begin{assumption}
\label{ass:SMOOTH}
Each potential outcome function $f^{A=a}(x)$, $a \in \{0, 1\}$, is Lipschitz continuous with Lispchitz constant, $L$
\end{assumption}
\section{Balance in existing designs}
\label{sec:related_work}

There have been a plethora of design procedures that attempt to explicitly improve balance. 
These approaches fall into three primary camps:
\begin{itemize}[nosep,leftmargin=1em,labelwidth=*,align=left] 
    \item \textbf{Blocking}~\citep{greevy2004optimal, higgins2016improving}.
    Units are divided into a partition and then a fixed number of units are randomly given treatment within each stratum.
    This ensures that treatment is balanced on stratum indicators.
    \item \textbf{Rerandomization}~\citep{morgan2012rerandomization, li2018asymptotic}.
    Units are assigned completely randomly to treatment, then balance is checked.
    If imbalance is too high, then randomization is performed again.
    This process is repeated until imbalance is below some a priori specified level.
    \item \textbf{Optimization}~\citep{kallus2018optimal}.
    An optimization procedure is used to find the best vector of assignments to treatment in order to minimize some measure of imbalance. 
    This assignment may be deterministic.
\end{itemize}
These approaches can be difficult to scale to the necessary sample sizes for the online environment, as finding optimally balanced treatment assignments is an NP-hard problem.

The optimization objective most commonly employed for optimal experimental design is mean balance~(c.f. \citet{morgan2012rerandomization, kallus2018optimal}), i.e., minimizing the distance in means between the instances of $\mathbf{X}$ that are allocated to treatment and control, respectively.
This measure can be extended to incorporate higher order and non-linear dependencies by applying a feature transformation, $\phi$, to the covariates.
The resulting optimization problem is then given by 
\begin{align*}
    \min_{\mathbf{a} \in \{0, 1\}} g\left(\mathbf{a} \phi(\mathbf{X}_i), (1 - \mathbf{a}) \phi(\mathbf{X}_i)\right)
\end{align*}
where $g(\cdot)$ is a distance function. 
Popular choices for $g(\cdot)$ are Euclidean~\citep{hansen2008covariate}, and Mahalanobis~\citep{morgan2012rerandomization} distance.
Of particular interest to this work is the balance measure used by \citet{kallus2018optimal} which considers the mean difference between $\mathbf{X}$ after projecting the covariates in to a reproducing kernel Hilbert space~(RKHS). 
Defining $K$ to be the Gram matrix corresponding to 
the optimal experimental design under corresponds to solving the following binary quadratic program, termed the pure strategy optimal design~(PSOD) by its author,
\begin{align}
    \label{prob:PSOD}
    \min_{\mathbf{u} \in \{-1, 1\}} \frac{4}{N^2}u^T \mathbf{K} u
\end{align}
where $\mathbf{K}$ is the Gram matrix of $\mathbf{X}$ with respect to an RKHS. 
Under smoothness assumptions on the potential outcome function, the solution to PSOD was shown to be Bayes optimal, with variance guarantess comparable to those provided by post-hoc regression adjustment.

While mean balance has intuitive and theoretical appeal, it also comes with significant computational disadvantages. 
\citet{kallus2018optimal} shows that PSOD, which accomodates a large number of mean balance measures, is equivalent to solving the balanced number partition problem which is known to be NP-hard. 
The implemented solution requires solving a semi-definite program which prevents the applicability of the method to moderately large domains~(in the hundreds to thousands).

\section{Experimental Design for ITE}
\label{sec:ite_design}
Design for the average treatment effect has received considerable attention in the literature.
Less studied, however, is design specifically targeting individual treatment effects. 
Recently, this quantity has gained substantial attention due to~\citet{athey2016recursive,wager2018estimation} and the broader literature around individual treatment effect estimation~\citep{shalit2017estimating, shi2019adapting}\footnote{In the statistics and econometrics literature, this task is often referred to as ``heterogeneous treatment effect'' estimation.}.
The ITE is given by the difference in potential outcomes conditioned on $x$, $Y^{A=0}(x) - Y^{A=1}(x)$, $x \in X$.
The central task considered within this paper is allocating treatment to estimate the ITE well (we will make this statement more formal shortly). 

To motivate our design task, we begin with an estimator for the individual treatment effects.
We will restrict ourselves to distance based regression functions,
\begin{align}
\label{eq:regressor}
\hat{f}(x_i) = \sum_j^N w_{ij} y_j & \quad \quad s.t.
\sum_j^N w_{ij} = 1
\end{align}
Special cases of this general formulation are k-nearest neighbors regression as well as Nadaraya-Watson kernel regression.
These estimators are non-parametric and fairly flexible.
We focus on this estimator due to its analytical tractability in combination with its generally reasonable performance as a non-parametric estimator.
As we will show in section~\ref{sec:experiments}, a design which is effective for this estimator will typically also be effective for other ITE estimators.
Under the assumed model, the empirical estimate of the ITE can be written as
\begin{align}
\label{eq:ite}
\hat{\tau}_i = (2a_i - 1)\left(y_i - \sum_{j=1}^n w_{ij} y_j\right).
\end{align}
Equation \ref{eq:ite} can be interpreted as two independent regressions inferring the potential outcomes of $Y^{A=1}$ and $Y^{A=0}$, where the predictions for observed potential outcomes are constrained to be equal to the observed outcome.
In the individual treatment effect estimation literature, training an outcome for each potential outcome surface is often referred to as a ``T-learner''~\citep{kunzel2019metalearners}, but due to our restriction that observed potential outcomes take their observed values, our approach is more similar to the ``X-learner'' of~\citet{kunzel2019metalearners}.
This restriction is also often employed in the tranductive learning setting, for example, by~\citet{zhu2003combining}.
Framed in terms of transductive inference, our task of ITE estimation is to impute the counterfactual for each unit, and this imputation of counterfactuals is the only way that error is introduced into our estimation problem.

To our knowledge, this paper is the first to examine designing an experiment explicitly for the estimation of ITEs. 
We do so by viewing experimental design as an optimal graph cut problem.
We discuss the details of the connection between graph cutting and experimental design for ITE estimation next.

\subsection{Graph cutting in experiments}
\label{sec:graph_cutting}
A natural interpretation of the assignment problem is to view the observations of covariates as nodes in a graph with with treatment being a missingness indicator.
Through this lens we see that the task of treatment assignment can be interpreted as minimizing the risk of two interrelated regression problems: predicting the control counterfactual for treatment using only control units, and predicting the treated counterfactual for control using only the treated units. 
The resulting optimization problem is then given as
\begin{align}
\label{prob:opt}
    \min_{A} \sum_i^N \left| {\sum_j  \frac{\mathbbm{1}(a_i\neq a_j)e_{i,j}}{\sum_k \mathbbm{1}(a_i\neq a_j)e_{i,k}} y_j} - y_i \right|
\end{align}
where we refer to similarity between points as $e_{i,j}$ and replace $w_{i,j}$ with a more explicit expression.
Note that the choice of similarity function, as before, is a design choice made by the practitioner.
As with most causal inference applications, the outcomes are unobserved which can make reasoning over design choices difficult a priori.  
However, after leveraging the Lipschitz assumption (assumption \ref{ass:SMOOTH}) the following proposition allows for a bound on the bias and variance of the regression function.

\begin{proposition}\citep{anava2016k}
\label{prop:knnbound}
The bias and variance of an estimate for any one point, $x^\star$ is bounded by
\begin{align}
 \nonumber
&\left|\sum_{i=1}^{n} w_{\star, i} y_{i}-f\left(x_\star\right)\right|
\leq
\underbrace{C\|\mathbf{w}_\star\|_2}_{\text{variance}} + \underbrace{L\sum_{i=1}^{n} w_{\star, i} d\left(x_{i}, x_{\star}\right)}_{\text{bias}}
\end{align}
With probability $1 - \delta$
Where $L$ is the Lipschitz constant as in assumption \ref{ass:SMOOTH},  $d$ is a distance measure,  $C = b\sqrt{2\log(\frac{2}{\delta})}$, and $b$ upper bounds noise,  i.e., $|\epsilon| \leq b$.
\end{proposition}
A proof is provided in the supplement for completeness.
Proposition \ref{prop:knnbound} provides an expression for the error which relies only on observable quantities, namely the distance between treatment and controls and the regression weights, and an assumption on the magnitude of noise.
The optimization problem in equation \ref{prob:opt} can then be recast as
\begin{align*}
    &\min_A \sum_j^n 
    C\|\mathbf{w}_j\|_2 + 
    {L\sum_{i=1}^{n} w_{j, i} d\left(x_{i}, x_{j}\right)}
\end{align*}
with $w_{j,i} = \frac{\mathbbm{1}(a_i\neq a_j)e_{i,j}}{\sum_k \mathbbm{1}(a_i\neq a_j)e_{i,k}}$ as in equation \ref{prob:opt}.
This lens makes explicit the tradeoffs between bias and variance in the design. 
It should come as no surprise that the optimal design will be heavily reliant on the distribution of $\mathbf{X}$ and the magnitude of the noise term, i.e. the size of $b$. 
For example, on one extreme when $b$ is close to zero, then the best choice will be to concentrate all of the weight on the first nearest neighbor. 
As we discuss in section \ref{sec:matchedpairs}, this corresponds to a greedy design which two-colors a one-nearest neighbor graph. 
More generally, it is necessary to reason over trade-offs that are occurring with respect to the experimental design.

In this work, we propose to view these choices by recasting the problem of experimental design in terms of graph cutting. 
Specifically, we consider a graph, $G$ where the edge weights, $e_{i,j}$ are the similarity between $\mathbf{x}_i$ and $\mathbf{x_j}$. 
After remapping treatment to $\{-1, 1\}$ via $\mathbf{u} = 2\mathbf{a} - 1$, the problem of treatment assignment can be recast as choosing an assignment.
This view is quite natural, since the set of cut edges, i.e., edges where $a_i \neq  a_j$ are those which are used to infer the counterfactuals in the nearest neighbor regression.
The following proposition makes this more formal by relating the risk of the regression estimator to the Maxcut problem
\begin{proposition}
\label{prop:main_bound}
\begin{align*}
\sum_i^N \left|\sum_j w_{i,j} y_j - y_i \right| \leq 
\sum_i \epsilon_i + \frac{e_{\text{sum}} - \sum_{i,j} \mathbbm{1}(a_i\neq a_j) e_{ij}}{d_{\min}}
\end{align*}
Where $e_{sum} = \sum_{i,j} e_{i,j}$, and $d_{\min} = \min_i \sum_j e_{i,j}$.
\end{proposition}
The proof is provided in the supplement.
The first term is an irreducible component which corresponds to the estimation error due to non-smoothness of the potential outcome function.
It shows the error under an oracle scenario in which the unobserved potential outcome for a unit is estimated based on the potential outcomes for \emph{all} other units.
It further presumes this estimation is performed for every unit, which is not possible due to the fundamental problem of causal inference.
This represents the Bayes risk of the estimation problem: the lower bound of the error incurred for this estimation.
The second term is more interesting, as it describes the error due to the assignment process we choose.
While $w_{\text{sum}}$ and $d_{\min}$ do not depend on the assignment (and therefore optimal design need not incorporate them), the remaining piece does.
This term is the negative of the objective of the Maxcut graph-cutting problem, which we now describe in greater detail.

\subsection{Maxcut}
First, informally: maxcut divides the nodes of a graph into two disjoint and exhaustive subsets by removing (``cutting'') edges with the maximum edge-weights.

A common way to write this is through the use of the graph Laplacian:
\begin{align}
    \max_{\mathbf{u} \in \{-1, 1\}} \mathbf{u}^T L \mathbf{u}
    \label{eq:qp_maxcut}
\end{align}
where $\mathbf{u}$ corresponds to which set each node belongs to, denoted $-1$ and $1$.
The graph Laplacian is a matrix which represents the structure of the network, formed as the diagonal matrix of node-degree minus the incidence matrix of edge weights, $D - G$.
Maxcut is a canonical NP-hard problem and, is not amenable to a polynomial-time approximation scheme unless the unique games conjecture is true~\citep{khot2007optimal, goemans1995improved}.
Common approximation algorithms include semidefinite programming~\citep{goemans1995improved, trevisan2012max}.
The best known approximation ratio for this problem, in general, is through semidefinite programming, with a ratio of $\frac{16}{17}\approx 0.941$.
Given the difficulty of this problem, it is not possible to uniquely minimize Proposition~\ref{prop:main_bound} in polynomial time, so we will focus only on efficient algorithms for the computation of a design.
Note that the kernel allocation procedure proposed by Kallus~\citep{kallus2018optimal} is isomorphic to the Maxcut problem. 
We provide a proof of this correspondence in the supplement.

Certain special cases, however, allow for efficient solutions to Maxcut.
Among these are particular bipartite graphs such as forests and trees.
These graphs, for instance, admit solutions to Maxcut in \emph{linear} time.

\subsection{Optimal Matched Pair Designs}
\label{sec:matchedpairs}
The results above can help shed light on common experimental designs, such as the matched pair design~\citep{imai2008variance}.

\citet{kallus2018optimal} demonstrates that, when outcomes are Lipschitz, implementing this design by finding the max weight matching is optimal.
This can be efficiently implemented using, e.g. the Edmonds' algorithm~\citep{edmonds1967optimum}.
This optimality result, however, restricts the set of designs to those which may be defined as a matching on the graph.
A graph matching, of course, may have no two edges which share an end-point.
Our result demonstrates that a wider class of designs may be considered, opening the door to stronger assignment mechanisms.

In the observational literature on matching methods, there is a distinction drawn between greedy and so-called ``optimal'' matching~\citep{stuart2010matching, hansen2006optimal}.
The distinction being that a greedy matching algorithm can ``double dip,'' using the same unit as the matched control for multiple treated units.
The experimental design based on optimal matching is the \citet{kallus2018optimal} matched pair design, but we can similarly form a greedy design by two-coloring the one-nearest neighbor graph.
The decomposition of Proposition~\ref{prop:knnbound} gives us a ground on which to compare these designs.
The greedy design ensures minimal pointwise bias for the ITE by minimizing the distance to a match.
While providing the minimum pointwise bias of the design, the variance properties are not so clearcut.
Depending on specific properties of the data, either the greedy design or the matched-pair design could be lower variance.

\subsection{On Optimal Designs}

We now turn to the question: can we construct a feasible ``optimal'' design?
To provide a specific example from the previous section, how should practitioners decide between greedy designs (which may imply higher leverage for certain observations) and non-greedy designs (which may imply higher bias for the imputation of ITEs for some units)?
This question does not have easy answers.
Indeed, a simple example can illustrate this conundrum.
Suppose a graph with one point in the center in two dimensions with $n-1$ points surrounding it in a circle.
Further suppose that each of these points is $r$ units away from the centerpoint, but $s > r$ units away from the next closest point on the exterior.
Assume that each unit has a residual, $\epsilon$ (as in Assumption~\ref{ass:additive}), which is drawn from a mean zero normal distribution with standard deviation $\sigma_c$ for the center point and $\sigma_e$ for exterior points.
The greedy design would ensure that the center and the exterior points received different treatments.
The matched pair design would pair the center with one random exterior point, and then match all other exterior points with a neighboring exterior point.
Then we can write the expectation of the bound in equation~\ref{eq:regressor} for the center point in the matched-pair design as $\sigma_e + L r$.
For one point in the exterior, that quantity is $\sigma_c + L r$, while for all others it is $\sigma_e + L s$.
For the greedy design, this quantity would instead be $\frac{\sigma_e}{n-1} + L r$ for the center point.
For all exterior points, the bound would be $\sigma_c + L r$.

Depending on the relative values of $\sigma_c$ versus $\sigma_e$ and the distance $s$ versus $r$, either the greedy design or the matched-pair design could minimize the bound.
That is, if $\sigma_c$ is very large, then the greedy design will tend to exhibit variance properties that overwhelm its low bias.
Similarly, if $\sigma_e$ tends to be larger (or $s \gg r$), then the matched-pair design will tend to have unacceptably large biases that will overwhelm its variance properties.
Of course, in the asymptotic regime, only bias matters and thus the greedy design will minimize this bound.
In finite samples, this thoroughly unsatisfying bias-variance tradeoff demonstrates that the optimal design depends crucially on properties of the data which are unknowable a priori.
In short, we do not seek an optimal design, but instead simply designs that make a reasonable tradeoff between bias and variance for many applied situations.

Practitioners who understand more about their data, such as the extent of heteroskedasticity and the smoothness of the conditional expectation function can therefore make better decisions about design than any overarching theoretical statement that we can provide here.

\section{Novel Designs through Cutting Spanning Trees}
\label{sec:novel_design}
We begin by limiting our space of algorithms to scenarios in which Maxcut can be efficiently solved.
Since trees and forests admit linear-time solutions to Maxcut, we focus on them.

In proposition~\ref{prop:main_bound}, it is clear that integrated absolute bias is minimized when, for each unit, the similarity to the units with positive weights (i.e. the impute counterfactual) are maximized, as discussed in Section~\ref{sec:matchedpairs}.
The easiest way to ensure this is to match each unit with its closest neighbor in the graph and ensure that each neighbor in this newly sparsified graph receives different treatments than its neighbors.
This solution is where we begin; we call this design ``GreedyNeighbors'', because the design is realized by solving Maxcut \emph{exactly} on the one-nearest-neighbor graph.
The nearest neighbor graph can be computed efficiently in $\mathcal{O}(n \log n)$ time by using a $kd$-tree.
The one-nearest-neighbor graph is a forest, so solving Maxcut is trivially accomplished in $\mathcal{O}(n)$ by greedily walking the forest alternating treatment assignment.
Thus, this design is realizable in aggregate $\mathcal{O}(n\log n)$ time.
An important thing to note about this design in contrast to typical matched-pair designs is that a unit may be ``matched'' to more than one unit. 
Note that there are many realizable assignments with the GreedyNeighbors design, as each disconnected subgraph of the nearest-neighbor graph is assigned independently.
This implies that there are $2^M$ possible assignments, where $M$ is the number of disconnected subgraphs of the nearest-neighbor graph.

\begin{algorithm}
\SetAlgoLined
\SetKwInOut{Input}{input}
\SetKwInOut{Output}{output}
\Input{$X \in \mathbb{R}^D$}
\Output{Assignments $A$, Spanning Tree $T$}
 $G \gets \textrm{Similarity matrix constructed from }X$\\
 $T \gets \textrm{Maximum Spanning Tree}(G)$\\
$A \gets \textrm{MAXCUT}(T)$
 \caption{Deterministic Friedman-Rafsky Minimizing Design}
 \label{alg:derministicdesign}
\end{algorithm}

This design, however, despite minimizing bias on the ITE estimates, is needlessly high variance.
Each added edge will stabilize the variance component in the decomposition in proposition~\ref{prop:knnbound}.
Thus, adding edges will reduce the variance of the ultimate solution (at the expense of some additional possibility for bias).
We propose a design which manages this tradeoff in a computationally tractable way based on the maximal spanning tree of the original similarity graph.
Algorithm~\ref{alg:derministicdesign} summarizes this design.
In short, the maximal spanning tree (MST) is the largest tree over the graph which contains no loops or cycles.
The maximal spanning tree always contains the nearest neighbor graph (as in, all edges of the nearest neighbor graph also are within the maximal spanning tree).
Since the MST is a tree, it can also be solved trivially by Maxcut in $\mathcal{O}(n)$.
The MST itself can be computed in $\mathcal{O}(n \log n)$.
Thus, the full procedure requires, again, only $\mathcal{O}(n \log n)$ time complexity.
Adding any additional edge to the MST which fails to preserve the bi-partiteness of the graph will make it no longer amenable to a greedy solution to Maxcut.
This makes it the largest graph (in terms of total edge-weight), for which Maxcut is \emph{necessarily} able to be efficiently solved.
We refer to this algorithm as ``SoftBlock'', since it softens the idea of a blocked design by allowing for substantial correlations between any two units (rather than simply units which lie within the same block).

\subsection{Probabilistic Interpretation}
\label{subsec:probabilistic_design}
We now provide a probabilistic interpretation of the proposed design. 
Starting with the observation that the set of all random spanning trees defines a determinantal point process~(DPP) where the probability of a spanning tree is proportional to the product of its edge weights~\citep{lyons2017probability}, i.e. $p(T) \propto \prod_{i,j \in E(T)} e_{i,j}$.
This can be trivially modified to represent a distribution where each tree's probability is given by its respective balance by first considering an exponentiation of the weights, i.e.,
$
p(T) \propto \prod_{i,j \in E(T)} \exp\left(e_{i,j}\right) = \exp\left(\sum_{i,j \in E(T)} e_{i,j}\right).
$

It's easily observed that when the sum of the weights is maximized (that is, the MST), the probability of the tree is also maximized.
Thus, SoftBlock, the design based on the MST, is the MAP estimate from this DPP.

\subsection{Balance and Graph Two Sample Tests}
All designs we have considered correspond to a particular test of balance between treated and control units.
For example, rerandomization using the Mahalonobis distance minimizes a $t$-test, and as we detail in the appendix, problem \ref{prob:PSOD} corresponds to minimizing an uncentered version of maximum mean discrepancy~\citep{gretton2012kernel}. 

As it turns out, SoftBlock shares an interesting connection to the minimum spanning tree test of \citet{friedman1979multivariate} 
Specifically, the graph based test addresses the problem of detecting differences between two distributions by viewing the problem in terms of a cut on a minimum spanning tree.
The procedure is as follows. 
The two samples $\mathbf{X}_0$, and $\mathbf{X}_1$ are pooled and a similarity graph, $\mathcal{G}$ is constructed according to an analyst specified similarity metric. 
The minimum spanning tree, $\mathcal{T}$ for $\mathcal{G}$ is then found. 
The test statistic is defined as the number of edges in $\mathcal{T}$ that connect samples from $\mathbf{X}_0$ and $\mathbf{X}_1$, i.e., $\text{FR}(\mathbf{X}_0, \mathbf{X}_1) = \sum_{\mathbf{x}_i \in \mathbf{X}_0, \mathbf{x}_j \in \mathbf{X}_1} \frac{\mathbbm{1}(\mathbf{x}_i, \mathbf{x}_j) \in E(\mathcal{T})}{N - 1}$, where $E(\mathcal{T})$ are the set of edges in the minimum spanning tree, $\mathcal{T}$, and $N$ are the total number of samples in the pooled dataset. 
 The test is minimized if the two samples share only one edge in the spanning tree, and maximized when edges connect units from different samples as much as possible.
This procedure was shown to be asymptotically normal and consistent by \citet{henze1999multivariate}.
The SoftBlock assignment mechanism directly minimizes the Friedman-Rafsky test statistic.
By optimizing a consistent test of balance, our procedure asymptotically guarantees balance on covariates between groups. 
Given that this is a consistent test, we can be sure that even though we aren't directly optimizing linear balance, we will converge to linear balance in the limit.
In finite samples, this procedure may sacrifice some degree of linear balance relative to traditional blocking procedures.
Essentially, linear balance implies a computationally intractable (i.e. NP-hard) exact solution, while the use of a different metric of balance provides a simple polynomial time algorithm (with equivalence to the linear problem in the limit).
\begin{figure}
\begin{center}
    \includegraphics[width=0.7\textwidth]{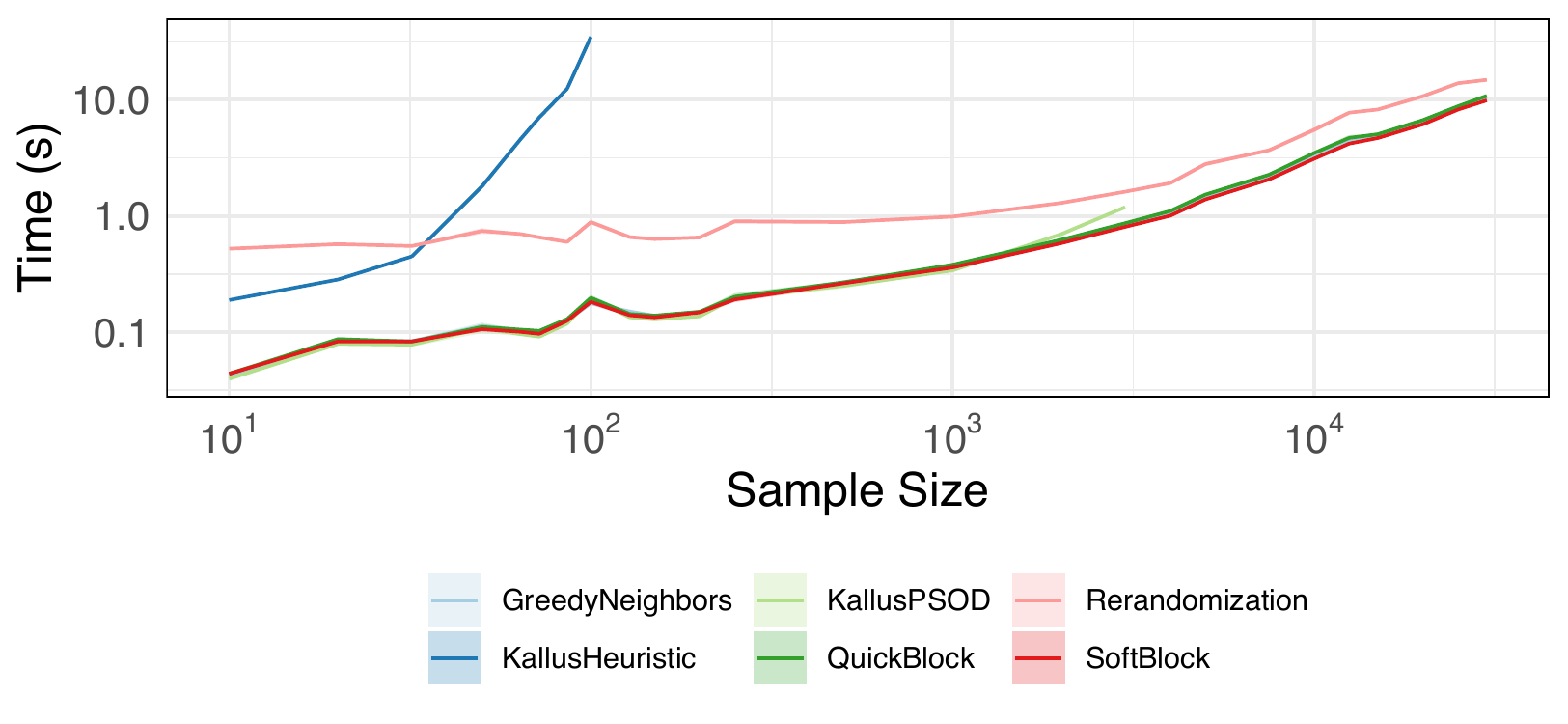}
    \caption{
    Runtime of various methods for experimental design. 
    Both axes are logarithmic.
    }
    \label{fig:runtime}
\end{center}
\end{figure}

\section{Experiments}
\label{sec:experiments}
In this section, we present experiments demonstrating the effectiveness of SoftBlock.
We begin by describing the methods we benchmark against:
\begin{itemize}[nosep,leftmargin=1em,labelwidth=*,align=left] 
\item Bernoulli randomization. This method flips a fair coin for each unit.
This method is minimax optimal for the ATE as per \citet{kallus2018optimal}.
\item Rerandomization. The method of \citet{morgan2012rerandomization} randomizes, checks balance (by Mahalanobis distance) and, if it's too high, repeats. In our implementation, we use the heuristic of \citet{kallus2018optimal}, which accepts a randomization with only 1\% probability.
Thus, it ensures that the chosen design has one of the 1\% most balanced designs (in terms of Mahalanobis distance).
\item QuickBlock. The method of \citet{higgins2016improving} finds an approximate blocking solution based on a $k$ nearest neighbor graph. They find that it performs comparably or better to \citet{greevy2004optimal} ``optimal'' blocking.
\item Kallus' PSOD and Heuristic Designs. These designs of \citet{kallus2018optimal} optimize assignments to minimize mean imbalance in an RKHS.
\item Optimal matched pair designs. These designs solve the maximum weight matching problem and then randomize which unit in each pair receives treatment~\citep{kallus2018optimal, imai2008variance}.
\end{itemize}

We consider a variety of linear and non-linear data generating processes, defined in detail in Table~\ref{tab:dgps} in the Appendix.
The QuickBlockDGP was the primary simulation used in \citet{higgins2016improving}, consisting of four uniform random variables multipled together.
We additionally provide a simulation with a linear outcome, one based on a sinusoid, and one with covariates distributed along two circumscribed circles.

Figure~\ref{fig:runtime} shows the runtimes of these various methods on the TwoCircles data generating process.
At very low sample sizes, \citeauthor{kallus2018optimal}'s \citeyearpar{kallus2018optimal} PSOD method is the fastest way to design an experiment and estimate effects, but by moderate sample sizes is outpaced by QuickBlock, SoftBlock and the GreedyNeighbors methods.
SoftBlock is faster than QuickBlock at nearly all sample sizes, but the two approaches increase computational time at a similar rate.
\begin{figure}
    \centering
    \includegraphics[width=0.8\textwidth]{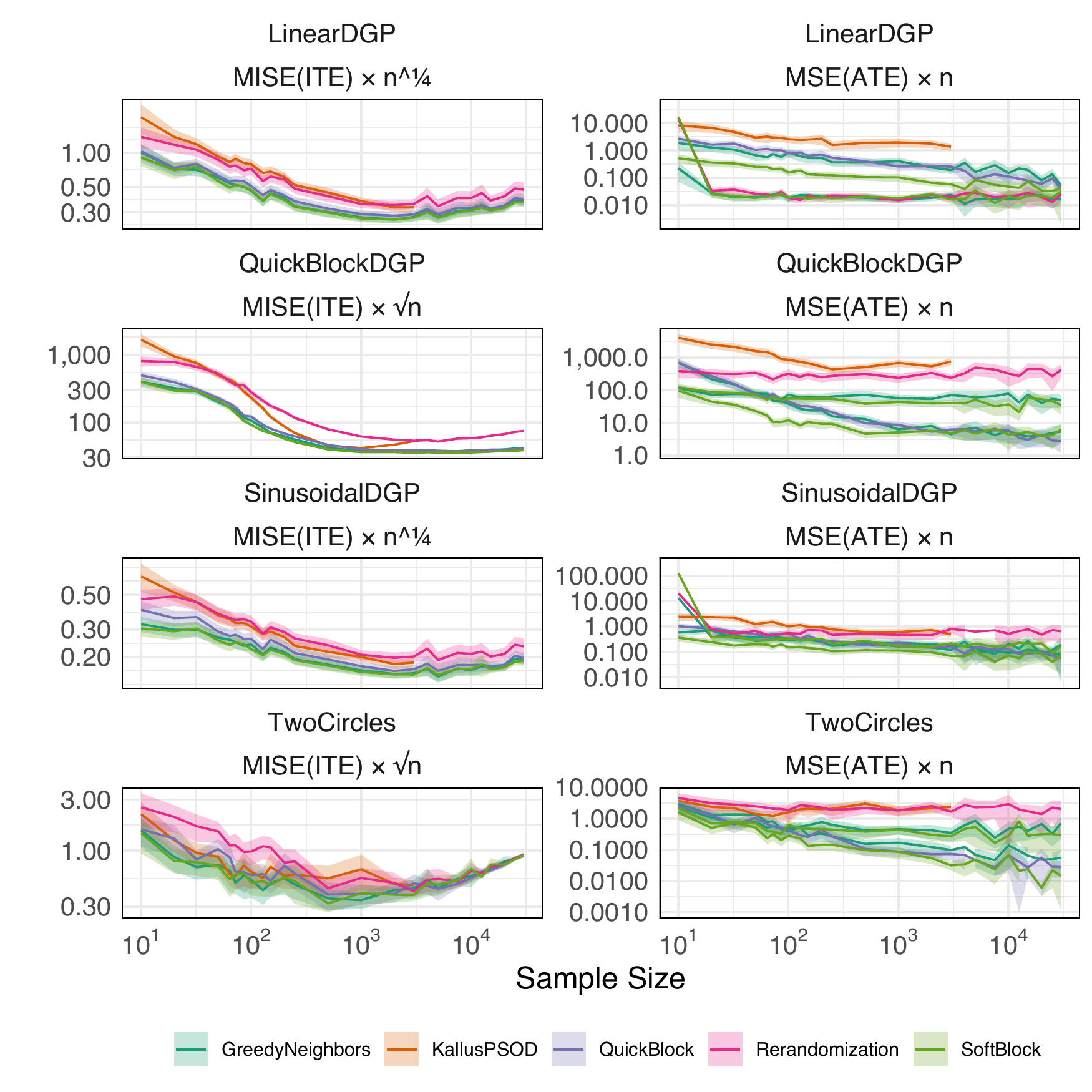}
    \caption{
    SoftBlock performs well across a wide-array of data-generating-processes and sample sizes.
    These plots display mean squared-error and mean integrated-squared-error for the ATE and ITE, respectively. 
    These values are multiplied by a function of sample size for ease of comparison. 
    }
    \label{fig:performance_sims}
\end{figure}

Figure~\ref{fig:performance_sims} shows the performance of the methods on the simulation setups in Table~\ref{tab:dgps}.
Note that values in this chart are normalized for sample size (errors are multiplied by $n^{\sfrac{1}{K}}$ to allow for easier comparison across a wide array of sample sizes).
In the LinearDGP, the methods which estimate the ATE with \citeauthor{lin2013agnostic}'s \citeyearpar{lin2013agnostic} regression-adjustment method are, in fact, \emph{correctly specified} parametric models. 
As such, they (Rerandomization, Greedy Nearest Neighbors and Bernoulli randomization) have much lower error than competitor methods.
SoftBlock, however, converges to nearly the same error by around $n=10000$.
In the LinearDGP, SoftBlock and Greedy Nearest Neighbors are substantially more effective at estimating the ITEs than competitor methods, with SoftBlock outperforming Greedy Nearest Neighbors.
Similar patterns hold in terms of the ITE on all DGPs, with QuickBlock performing the closest to SoftBlock, particularly at higher sample sizes.
For estimating the ATE on the non-linear DGPs, SoftBlock is nearly always the most effective method, often substantially so, for example in moderate sample sizes on the QuickBlockDGP.
The comparison between the GreedyNeighbors design and SoftBlock is informative, since the MST always contains the nearest neighbor graph.
SoftBlock has two main advantages over this design.
First, it reduces variance by using more than just the closest neighbor (for instance, sometimes the two nearest neighbors are both very close, so it would be wise to use both of them).
Second, by being a single connected graph, it ensures that the assignments across different pairs of nearest neighbors are ``lined up''.
That is, it avoids certain bad randomizations, in which, for example, two nearby edges are oriented in the same direction wherein the unit with larger covariate value is assigned treatment in both pairs.
The cut on the MST, on the other hand, is more likely to insulate against this eventuality by connecting these subgraphs and ensuring the orientation of treatments do not match.

Figure~\ref{fig:design_ests} shows the performance of the design-based estimators for ITEs.
In contrast to the previous figure, which estimates ITEs with a random forest T-learner, this figure shows the ITEs estimated by only the specific estimator implied by the design.
This means that, for a blocking estimator, a difference-in-means estimator is used within each block to impute conditional effects (which are assumed constant within blocks).
For SoftBlock, the ITE estimator is the difference of the observed ego unit and its synthetic counterfactual constructed by the weighted average of its neighbors in the minimum-spanning-tree as analyzed in section~\ref{sec:graph_cutting}.
In this comparison, SoftBlock performs substantially and consistently better than other designs.
The comparison to blocking methods in this experiment demonstrate why SoftBlock is able to do better at estimating the ITE than other methods: it is optimized to ensure good interpolation across the entire space.
In particular, we can once again see as informative its comparative stability relative to the matched-pair designs (note that the \citet{kallus2018optimal} matched-pair design is infeasibly slow to display above sample sizes of 100 in this simulation).
\begin{figure}
    \centering
    \includegraphics[width=0.8\textwidth]{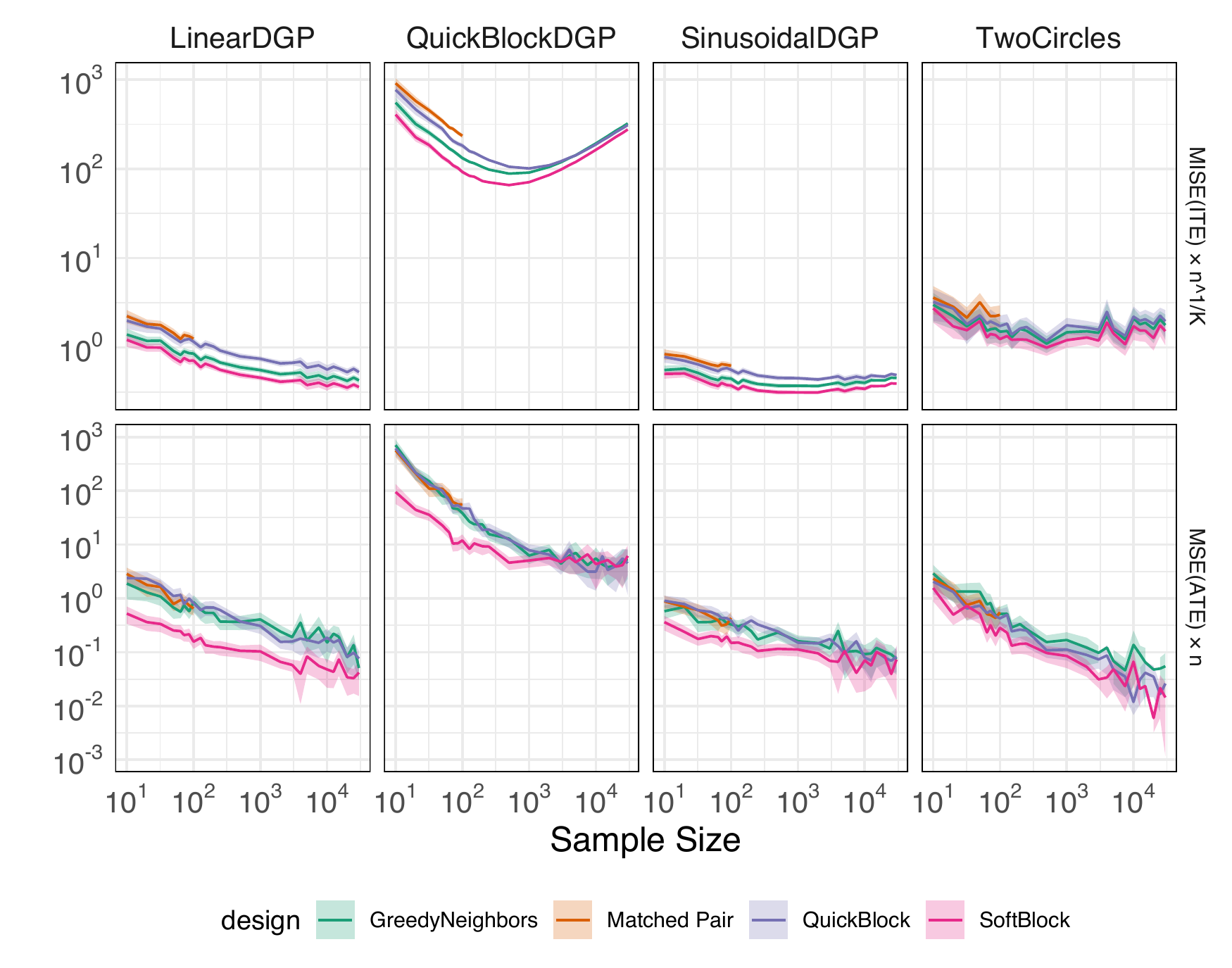}
    \caption{
    The design-based estimators from SoftBlock are appreciably better than existing methods.
    Mean squared-error and mean integrated-squared-error for the design-based estimators of the ATE and ITE, respectively. 
    These values are multiplied by a function of sample size for ease of comparison.
    }
    \label{fig:design_ests}
\end{figure}

\begin{figure}[t]
    \centering
    \includegraphics[width=0.8\textwidth]{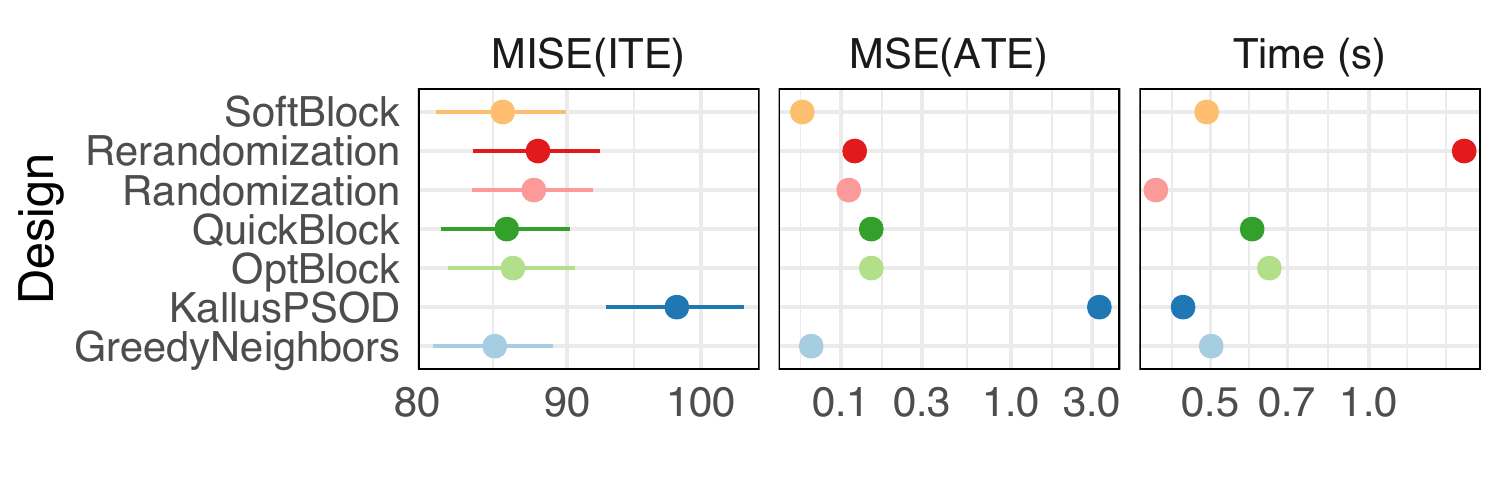}
    \caption{
    Mean squared-error of the ATE, mean integrated-squared-error of the ITE and time to calculate a design and perform estimation.
    }
    \label{fig:ihdp}
\end{figure}

Figure~\ref{fig:ihdp} shows the performance of various methods on the IHDP simulation study, as introduced by~\citet{hill2011bayesian}.
We compare using setting ``B'', in which the outcome model is nonlinear and the treatment effect is not constant.
In this data, SoftBlock provides the lowest error estimates of the ATE, and all methods tend to perform well for estimating the ITE with a random forest T-learner.

In the appendix, figure~\ref{fig:hyperparams} shows the sensitivity of the \citet{kallus2018optimal} methods and SoftBlock to hyperparameters (SoftBlock is very robust, while the \citet{kallus2018optimal} PSOD method only performs comparably when hyperparameters are chosen optimally).

\section{Conclusion}

In this paper, we've provided a framework through which to think about designs for individual treatment effect estimation and provided a formulation of the problem as graph cutting.
Through this framework we presented two novel experimental designs which are well-suited to estimating ITEs and compare them to prior work.
Simulations demonstrate that this method provides an improvement in terms of both computational tractability as well as efficiency.